\documentclass[10pt,a4paper]{article}
\pdfoutput=1
\usepackage[margin=1.5in]{geometry}
\usepackage{amsmath}
\usepackage[utf8]{inputenc}
\usepackage{physics}
\usepackage{amsthm}
\usepackage{amssymb}
\usepackage{graphicx}
\usepackage{float}
\usepackage{bbold}
\usepackage{xcolor}
\usepackage{colortbl}
\usepackage[sorting=none]{biblatex}
\bibliography{biblio.bib}

\usepackage{siunitx}
\usepackage{awesomebox}
\usepackage[affil-it]{authblk}  
\usepackage{mathtools}  
\usepackage{hyperref}
\usepackage{cleveref}
\hypersetup{final}

\newcommand{\rrho}{\boldsymbol{\rho}}
\newcommand{\llambda}{\boldsymbol{\lambda}}

\newcommand{\Amat}{\mathbf{A}}
\newcommand{\Bmat}{\mathbf{B}}
\newcommand{\Mmat}{\mathbf{M}}
\newcommand{\Hmat}{\mathbf{H}}
\newcommand{\Pmat}{\mathbf{P}}
\newcommand{\blmat}{\mathbf{b}}

\newcommand{\Dmat}{\mathbf{D}}
\newcommand{\Imat}{\mathbb{1}}
\newcommand{\Umat}{\mathbf{U}}

\newcommand{\Sblmat}{\mathbf{S}}

\DeclareMathOperator{\ssin}{s}
\DeclareMathOperator{\ccos}{c}
\DeclareMathOperator\Ln{Ln}

\newcommand{\SUN}[1]{\ensuremath{\text{SU}({#1})}}
\newcommand{\UN}[1]{\ensuremath{\text{U}({#1})}}

\newcommand{\sun}[1]{\ensuremath{\mathfrak{su}({#1})}}

\newcommand{\un}[1]{\ensuremath{\mathfrak{u}({#1})}}

\newtheorem{theorem}{Theorem}[section]

\newcounter{example}[section]

\title{Geometric invariant decomposition of \SUN{3}}
\author{Martin Roelfs$^a$\thanks{\href{mailto:martin.roelfs@kuleuven.be}{martin.roelfs@kuleuven.be}} \thanks{The research of Martin Roelfs~is supported by KU Leuven IF project C14/16/067.}}
\affil{\footnotesize $^{a}$ KU Leuven Campus Kortrijk--Kulak, Department of Physics, Etienne Sabbelaan 53 bus 7657, 8500 Kortrijk, Belgium}

\date{}

\begin{document}

\maketitle

\begin{abstract}
    A novel invariant decomposition of diagonalizable $n \times n$ matrices into $n$ commuting matrices is presented. This decomposition is subsequently used to split the fundamental representation of \sun{3} 
    Lie algebra elements into at most three commuting elements of \un{3}. 
    As a result, the exponential of an \sun{3} Lie algebra element can be split into three commuting generalized Euler's formulas, or conversely, a Lie group element can be factorized into at most three generalized Euler's formulas. After the factorization has been performed, the logarithm follows immediately.
\end{abstract}

\section{Introduction}

The aim of this paper is to identify the quantities left invariant by a given \SUN{3} transformation, and to describe the role these quantities play as the generators of the transformation.
Consider a traceless skew-Hermitian $3 \times 3$ matrix $\Bmat$; an element of the Lie algebra $\sun{3}$. 
We will demonstrate that such a matrix can be decomposed into at most three commuting matrices $\blmat_i \in \un{3}$:
    \begin{equation}
        \Bmat = \blmat_1 + \blmat_2 + \blmat_3.
        \label{eq:bivector_split_mat}
    \end{equation}
The $\blmat_i$ are said to be \emph{simple} because $\blmat_i^2 = \lambda_i \Imat$, where $\lambda_i \in \mathbb{R}$, $\lambda_i \leq 0$. Defining $\beta_i := \sqrt{-\lambda_i}$, it is easily verified that $\hat{\blmat}_i := \blmat_i / \beta_i$ squares to $-\Imat$. 
Therefore each $\blmat_i$ can be normalized to behave like an imaginary unit.

Because the \SUN{3} element corresponding to $\Bmat$ is $\Umat = \exp\bqty{\Bmat}$, it follows from the commutativity of the $\blmat_i$ that the exponential can be split into the product of three generalized Euler's formulas:
    \begin{align}
        \Umat &= e^{\Bmat} = e^{\blmat_1}e^{\blmat_2}e^{\blmat_3} \\
        &= \prod_{i=1}^3 \bqty{\Imat \cos \beta_i + \hat{\blmat}_i \sin \beta_i},
    \end{align}
where $e^{\blmat_i} \in \UN{3}$, but the product $e^{\blmat_1}e^{\blmat_2}e^{\blmat_3} \in \SUN{3}$.
As each $\blmat_i$ is invariant under the transformation $\Umat \blmat_i \Umat^\dagger$, the decomposition of  \cref{eq:bivector_split_mat} is called the \emph{invariant decomposition} of $\Bmat$.

The logarithm of $\Umat_i = \exp\bqty{\blmat_i}$ is not unique, in the same way that the complex logarithm is not unique \cite{ablowitz_fokas_2003}. Following a similar strategy to complex analysis, we first define a \emph{principal logarithm}, for which $0 \leq \beta_i \leq \pi$.
To this end, let us define
    \begin{align}
        \ccos(\blmat_i) &:= \frac{1}{2} \pqty{\Umat_i + \Umat_i^\dagger} = \Imat \cos \beta_i  \\
        \ssin(\blmat_i) &:= \frac{1}{2} \pqty{\Umat_i - \Umat_i^\dagger} = \hat{\blmat}_i \sin \beta_i. \label{eq:ssin_ccos_simple}
    \end{align}
As $0 \leq \beta_i \leq \pi$ implies $0 \leq \sin \beta_i \leq 1$, only the cosine function has to be inverted to obtain the principal logarithm:
    \begin{align}
        \Ln{\Umat_i} &= \widehat{\ssin(\blmat_i)} \arccos(\frac{1}{3} \tr\bqty{\ccos(\blmat_i)}). \label{eq:logarithm_simple} \\
        &= \hat{\blmat}_i \arccos(\cos \beta_i).
    \end{align}
Because the sign information of $\ssin(\blmat_i)$ is not carried by $\sin \beta_i$, but rather by $\hat{\blmat}_i$, a full $2 \pi$ range for the principal logarithm is maintained.
Consequently, by factoring a group element $\Umat$ into $\Umat_1 \Umat_2 \Umat_3$, a principal logarithm for $\Umat$ follows directly, as
    \begin{equation}
        \Ln(\Umat)= \Ln(\Umat_1) + \Ln(\Umat_2) + \Ln(\Umat_3).
    \end{equation}
Closed forms for the exponential function of \sun{3} elements have been published before, see e.g. \cite{Curtright:2015iba,doi:10.1063/1.4938418}. However, the invariant decomposition of \cref{eq:bivector_split_mat} presents an intuitive approach, which is also easy to invert to give a closed form logarithm for \SUN{3} elements. 
Additionally, the invariant decomposition of \cref{eq:bivector_split_mat} gives the invariants $\blmat_i$ a very strong geometric interpretation, as the invariants of the transformations they generate.

This paper is organized as follows. Firstly, in \cref{sec:bivector_split_mat} we define the invariant decomposition of diagonalizable $n \times n$ matrices. Secondly, in \cref{sec:exponential_mat} we define the exponential function of \SUN{3} elements using the invariant decomposition. Thirdly, in \cref{sec:factorization_mat} we describe how the factorization $\Umat = \Umat_1 \Umat_2 \Umat_3$ is performed. Fourthly, in \cref{sec:logarithm_mat} we describe the logarithm of $\Umat \in \SUN{3}$. 
Lastly, we apply the invariant decomposition to the Gell-Mann matrices in \cref{sec:gellmann}.

\section{Methods}

\subsection{Invariant decomposition}\label{sec:bivector_split_mat}

We start by defining the invariant decomposition for $3 \times 3$ diagonalizable matrices in \cref{th:orthogonal_decomposition}.
After this, the $n \times n$ case follows directly.

\begin{theorem}\label{th:orthogonal_decomposition}
    A $3 \times 3$ diagonalizable matrix $\Bmat$, can be decomposed into at most three commuting normal matrices $\blmat_i$, satisfying $\blmat_i^2 = \lambda_i \Imat$. Here $\lambda_i \in \mathbb{R}$, $\lambda_i \leq 0$.
\end{theorem}
\begin{proof}
    As the $3 \times 3$ matrix $\Bmat$ is diagonalizable, it can be written as $\Bmat = \Pmat \Dmat \Pmat^{-1}$, where $\Dmat = \text{diag}(\alpha_1, \alpha_2, \alpha_3)$ is a diagonal matrix whose diagonal entries are the eigenvalues $\alpha_i \in \mathbb{C}$.
    Assuming a decomposition $\Bmat = \sum_{i=1}^3 \blmat_i$ into commuting $\blmat_i$ exists, it follows that all $\blmat_i$ would be simultaneously diagonalizable, and thus
        \begin{equation}
            \Dmat = \Pmat^{-1} \blmat_1 \Pmat + \Pmat^{-1} \blmat_2 \Pmat + \Pmat^{-1} \blmat_3 \Pmat.
        \end{equation}
    We then make the ansatz
    \begin{align}
        \Pmat^{-1} \blmat_1 \Pmat &= \tfrac{1}{2} (\alpha_1 - \tr\bqty{\Bmat}) \, \text{diag}( +1, -1, -1 ) \label{eq:diagonalization} \\
        \Pmat^{-1} \blmat_2 \Pmat &= \tfrac{1}{2} (\alpha_2 - \tr\bqty{\Bmat}) \, \text{diag}( -1, +1, -1 ) \notag \\
        \Pmat^{-1} \blmat_3 \Pmat &= \tfrac{1}{2} (\alpha_3 - \tr\bqty{\Bmat}) \, \text{diag}( -1, -1, +1 ) \notag,
    \end{align}
    which satisfies $\Bmat = \sum_{i=1}^3 \blmat_i$, $\blmat_i^2 = \tfrac{1}{4}(\alpha_i - \tr\bqty{\Bmat})^2 \Imat$, and $\comm{\blmat_i}{\blmat_j}~=~0$. Therefore the sought-after decomposition has been found.
    When the eigenvalues $\alpha_i$ are degenerate, this decomposition is no longer unique. 
\end{proof}
\noindent
We notice that a decomposition of this type exists for any diagonalizable $n \times n$ matrix $\Bmat$. If $n \geq 3$:
    \begin{align}
        \Pmat^{-1} \blmat_1 \Pmat &= \tfrac{1}{2} (\alpha_1 - \tfrac{1}{n - 2} \tr\bqty{\Bmat}) \, \text{diag}( +1, -1, \ldots, -1, -1 ) \label{eq:diagonalization_nxn} \\
        \Pmat^{-1} \blmat_2 \Pmat &= \tfrac{1}{2} (\alpha_2 - \tfrac{1}{n - 2} \tr\bqty{\Bmat}) \, \text{diag}( -1, +1, \ldots, -1 -1 ) \notag \\
        \vdots \notag \\
        \Pmat^{-1} \blmat_n \Pmat &= \tfrac{1}{2} (\alpha_n - \tfrac{1}{n - 2} \tr\bqty{\Bmat}) \, \text{diag}( -1, -1, \ldots, -1, +1 ) \notag.
    \end{align}
The proof follows along the same lines as that of  \cref{th:orthogonal_decomposition}. To investigate if all the properties of the invariant decomposition discussed in this paper are also valid for $n > 3$ is outside the scope of this paper, and will be the topic of future research.

Having established that any diagonalizable $3 \times 3$ matrix $\Bmat$ can be split into at most three commuting matrices $\blmat_i$, we remark that when $\Bmat$ is traceless and skew-Hermitian, and thus $\Bmat \in \sun{3}$, the values of $\lambda_i$ can alternatively be calculated as the roots of
    \begin{align}
        \mathbb{0} &= \pqty{\blmat_1^2 - \lambda_i \Imat }\pqty{\blmat_2^2 - \lambda_i \Imat }\pqty{\blmat_3^2 - \lambda_i \Imat } \\
        \implies 0 &= - \lambda_i^3 + \frac{1}{4} \tr\bqty{\Bmat^2} \lambda_i^2 - \frac{1}{4} \pqty{\frac{1}{4} \tr\bqty{\Bmat^2}}^2 \lambda_i + \bqty{\frac{\det(\Bmat)}{8}}^2.\label{eq:roots}
    \end{align}
When all $\lambda_i$ are distinct, the $\blmat_i$ are found by solving
    \begin{equation}
        \blmat_i = \bqty{\Bmat + \frac{1}{8 \lambda_i} \Imat \det\pqty{\Bmat}} \bqty{\Imat + \frac{1}{2 \lambda_i} \pqty{\Bmat^2 - \frac{1}{4} \Imat \tr\pqty{\Bmat^2}}}^{-1}.
        \label{eq:simple_u3}
    \end{equation}
Thus, when all $\lambda_i$ are distinct, the invariant decomposition can be performed without performing diagonalization, but at the cost of an inverse. \Cref{eq:simple_u3} is the matrix representation of the orthogonal decomposition of bivectors, given in \emph{Clifford algebra to geometric calculus} \cite{GeometricCalculus}.
However, the invariant decomposition of \cref{th:orthogonal_decomposition} can always be performed, even when the $\lambda_i$ are degenerate.

\subsection{Exponential of an \sun{3} element}\label{sec:exponential_mat}

Since elements of \SUN{3} can be written as $\exp\bqty{\Bmat}$, with $\Bmat \in \sun{3}$, we would like a simple and intuitive way to compute exponentials of $\sun{3}$ elements.
The invariant decomposition of \cref{th:orthogonal_decomposition} provides such a method, since
the exponential of $\Bmat \in \sun{3}$ follows straightforwardly after performing the decomposition of $\Bmat$ into $\{ \blmat_i \}$.
Define $\beta_i := \sqrt{-\lambda_i}$, where $\lambda_i = \frac{1}{3} \tr\bqty{\blmat_i^2}$.
Then an \SUN{3} group element $\Umat = \exp\bqty{\Bmat}$ becomes
    \begin{align}
        \Umat &= e^{\Bmat} = e^{\blmat_1} e^{\blmat_2} e^{\blmat_3} \\
        &= \prod_{i=1}^3 \bqty{\ccos(\blmat_i) + \ssin(\blmat_i)} \notag \\
        &= \prod_{i=1}^3 \bqty{\cos(\beta_i) \Imat + \frac{\blmat_i}{\beta_i} \sin(\beta_i)},\notag
        \label{eq:group_element_mat}
    \end{align}
where $\ccos(\blmat_i)$ and $\ccos(\blmat_i)$ were previously defined in \cref{eq:ssin_ccos_simple}.
We can also form a family of elements generated by the $\blmat_i$:
    \begin{equation}
        \Umat(\theta_1, \theta_2, \theta_3) = e^{\theta_1 \blmat_1}e^{\theta_2 \blmat_2}e^{\theta_3 \blmat_3}.
    \end{equation}
Each $\blmat_i$ is an invariant of $\Umat(\theta_1, \theta_2, \theta_3) \in \UN{3}$, as due to the commutativity of the $\blmat_i$,
    \begin{equation}
        \Umat(\theta_1, \theta_2, \theta_3) \blmat_i \Umat(\theta_1, \theta_2, \theta_3)^\dagger = \blmat_i.
    \end{equation}
Therefore, $\Bmat$ is invariant under the entire family of transformations $\Umat(\theta_1, \theta_2, \theta_3)$
    \begin{equation}
        \Umat(\theta_1, \theta_2, \theta_3) \Bmat \Umat(\theta_1, \theta_2, \theta_3)^\dagger = \Bmat,
    \end{equation}
and so is any other linear combination of the $\blmat_i$, i.e.
    \begin{equation}
        \Amat(A_1, A_2, A_3) = A_1 \blmat_1 + A_2 \blmat_2 + A_3 \blmat_3.
    \end{equation}
Given the invariant decomposition of $\Bmat$, three parameters $\theta_i$ determine the group elements $\Umat(\theta_1, \theta_2, \theta_3)$ which leave $\Bmat$ invariant, and $\{ A_i \blmat_i \}$ spans the three parameter invariant subspace of $\Umat(\theta_1, \theta_2, \theta_3)$. For $\Umat(\theta_1, \theta_2, \theta_3)$ to be an element of \SUN{3} it needs to satisfy the constraint $\tr \bqty{ \theta_1 \blmat_1 + \theta_2 \blmat_2 + \theta_3 \blmat_3} = 0$; and for $\Amat(A_1, A_2, A_3)$ to be an element of \sun{3} it has to satisfy $\tr \Amat = 0$. Therefore there are only two degrees of freedom in these scenarios.

\subsection{Factorization of \SUN{3} element}\label{sec:factorization_mat}

From the exponential map of \cref{sec:exponential_mat}, we know that $\Umat \in \SUN{3}$ can be written as $\Umat_1 \Umat_2 \Umat_3$, where $\Umat_i = \exp\bqty{\blmat_i}$ and $\blmat_i$ is \emph{simple}. 
So given $\Umat \in \SUN{3}$, how do we find the $\Umat_i$? Splitting $\Umat(\Bmat) = e^{\Bmat}$ into cosine and sine, gives
    \begin{align}
        \ccos(\Bmat) &:= \frac{1}{2} \bqty{\Umat + \Umat^\dagger} \\
        \ssin(\Bmat) &:= \frac{1}{2} \bqty{\Umat - \Umat^\dagger}.
    \end{align}
Using $\ccos(\Bmat)$ and $\ssin(\Bmat)$, the \emph{grades} of $\Umat$ are defined as
    \begin{alignat}{3}
        \expval{\Umat}_0 &:= \ccos(\blmat_1)\ccos(\blmat_2)\ccos(\blmat_3) &&= \tfrac{1}{4} \Imat + \tfrac{1}{4} \tr\bqty{\ccos(\Bmat)} \Imat \label{eq:mat_grade0} \\
        \expval{\Umat}_2 &:= \sum_{i=1}^3 \ssin(\blmat_i) \prod_{j \neq i} \ccos(\blmat_j) &&= \ssin(\Bmat) - \expval{\Umat}_6 \label{eq:mat_grade2} \\
        \expval{\Umat}_4 &:= \sum_{i=1}^3 \ccos(\blmat_i) \prod_{j \neq i} \ssin(\blmat_j) &&= \ccos(\Bmat) - \expval{\Umat}_0 \label{eq:mat_grade4} \\
        \expval{\Umat}_6 &:= \ssin(\blmat_1)\ssin(\blmat_2)\ssin(\blmat_3) &&= \tfrac{1}{4} \tr\bqty{\ssin(\Bmat)} \Imat \label{eq:mat_grade6}.
    \end{alignat}
It is important to note that $\expval{\Umat}_2$ is not traceless, though it closely resembles the traceless projection commonly used in lattice Quantum Chromodynamics~\cite{doi:10.1142/6065,Mandula:1987rh,Giusti:2001xf}:
    \begin{equation}
        \ssin(\Bmat)\big\vert_\text{traceless} = \ssin(\Bmat) - \tfrac{1}{3} \tr \bqty{ \ssin(\Bmat) } \Imat.
        \label{eq:traceless_proj}
    \end{equation}
However, as $\tr\Bmat = 0$, it follows that in general $\tr \expval{\Umat}_2 \neq 0$: it contains a contribution proportional to the diagonal \un{3} generator $i \Imat$.
But because the generator $i \Imat$ is identical to the matrix representation of the pseudoscalar $i \Imat$, it is overzealous to discard the entire trace: only the part corresponding to the pseudoscalar $\expval{\Umat}_6$ has to be subtracted to obtain $\expval{\Umat}_2$.

An invariant decomposition of $\Amat = \expval{\Umat}_2 + \expval{\Umat}_4$ results in complex eigenvalues $\alpha_i=\gamma_i + i \delta_i$, where plugging only the real part $\gamma_i$ in \cref{eq:diagonalization} yields the Hermitian matrices
    \begin{equation}
        \Hmat_i = \ccos(\blmat_i) \prod_{j \neq i} \ssin(\blmat_j),
    \end{equation}
while the imaginary part $i \delta_i$ yields the skew-Hermitian matrices
    \begin{equation}
        \Sblmat_i = \ssin(\blmat_i) \prod_{j \neq i} \ccos(\blmat_j).
    \end{equation}
This yields a decomposition of $\Umat$ into its $8$ invariants:
    \begin{table}[H]
        \centering
        \begin{tabular}{c c c c}
             & ${\color{red}\ssin(\blmat_1)}{\color{blue}\ccos(\blmat_2)\ccos(\blmat_3)}$ & ${\color{blue}\ccos(\blmat_1)}{\color{red}\ssin(\blmat_2)\ssin(\blmat_3)}$ \\
             \\
            ${\color{blue}\ccos(\blmat_1)\ccos(\blmat_2)\ccos(\blmat_3)}$ & ${\color{blue}\ccos(\blmat_1)}{\color{red}\ssin(\blmat_2)}{\color{blue}\ccos(\blmat_3)}$ & ${\color{red}\ssin(\blmat_1)}{\color{blue}\ccos(\blmat_2)}{\color{red}\ssin(\blmat_3)}$ & ${\color{red}\ssin(\blmat_1)\ssin(\blmat_2)\ssin(\blmat_3)}$\\
            \\
             & ${\color{blue}\ccos(\blmat_1)\ccos(\blmat_2)}{\color{red}\ssin(\blmat_3)}$ & ${\color{red}\ssin(\blmat_1)\ssin(\blmat_2)}{\color{blue}\ccos(\blmat_3)}$
        \end{tabular}
    \end{table}
\noindent
Therefore there are a number of equivalent ways to perform the factorization. We define the matrix normalization procedure as
    \begin{equation}
        \hat{\Mmat}_i := \frac{\Mmat_i}{\norm{\Mmat_i}}, \qquad \norm{\Mmat_i} := \sqrt{\frac{1}{3}\tr\bqty{\Mmat_i \Mmat_i^\dagger}}. \label{eq:mat_normalize}
    \end{equation}
Then, the equivalent ways of calculating e.g. $\Umat_1$ are, up to the normalization of \cref{eq:mat_normalize},
    \begin{align}
        \Umat_1 &\propto \expval{\Umat}_0 + \Sblmat_1 \label{eq:factor_simple} \\
        &\propto \Imat + \Hmat_3 \Sblmat_2^{-1} \propto \Imat + \Hmat_2 \Sblmat_3^{-1} \label{eq:factor_inv} \\
        &\propto \Imat + \expval{\Umat}_6 \Hmat_1^{-1}. \label{eq:factor_inv2}
    \end{align}
The other $\Umat_i$ are obtained in analogous fashion.
In order to maintain the distinction between $\pm \Umat$, it is recommended to calculate only e.g. $\Umat_1$ and $\Umat_2$, after which the last factor follows as $\Umat_3 = \Umat_1^\dagger \Umat_2^\dagger \Umat$.

When $\expval{\Umat}_0 \neq 0$, it is computationally most efficient to use \cref{eq:factor_simple}, as \cref{eq:factor_inv,eq:factor_inv2} come at the cost of an inverse. However, this cost becomes unavoidable when $\expval{\Umat}_0 = 0$.

\subsection{Logarithm of an \SUN{3} element}\label{sec:logarithm_mat}

With the factorization $\Umat = \Umat_1 \Umat_2 \Umat_3$ of \cref{sec:factorization_mat} in hand, the principal logarithm is simply 
    \begin{equation}
        \Ln \Umat = \Ln \Umat_1 + \Ln \Umat_2 + \Ln \Umat_3,
    \end{equation}
where $\Ln \Umat_i$ is given by
    \begin{align}
        \Ln{\Umat_i} &= \widehat{\ssin(\blmat_i)} \arccos(\frac{1}{3} \tr\bqty{\ccos(\blmat_i)}) \\
        &= \hat{\blmat}_i \arccos(\frac{1}{3} \Re \tr\bqty{\Umat_i}).
    \end{align}
The logarithm is by no means unique. The principal logarithm $\Ln \Umat_i$ is one such logarithm, but so is
    \begin{equation}
        \ln \Umat_i = \Ln \Umat_i + 2 \pi k_i \hat{\blmat}_i,
    \end{equation}
where $k_i \in \mathbb{Z}$. It follows that all possible logarithms of $\Umat \in \SUN{3}$ are given by
    \begin{equation}
        \ln \Umat = \sum_{i=1}^3 \Ln \Umat_i + 2 \pi k_i \hat{\blmat}_i.
    \end{equation}
Each of the logarithms $\ln \Umat_i$ behaves just like the complex logarithm of complex analysis \cite{ablowitz_fokas_2003}, and so the theory of complex analysis can be brought to bear on studying their properties. 

\subsection{Decomposition of Gell-Mann matrices}\label{sec:gellmann}
The Gell-Mann matrices are the Hermitian generators of \sun{3}, and play an important role in Quantum Chromodynamics \cite{doi:10.1142/6065,GellMann:1961ky,peskin1995introduction}. It is therefore important to discuss how they decompose under \cref{th:orthogonal_decomposition}. Consider the Gell-Mann matrix $\llambda_1$:
    \begin{equation}
        \llambda_1 = \mqty({0 & 1 & 0 \\ 1 & 0 & 0 \\ 0 & 0 & 0}) = \tfrac{1}{2} \mqty({0 & 1 & 0 \\ 1 & 0 & 0 \\ 0 & 0 & 1}) + \tfrac{1}{2} \mqty({0 & 1 & 0 \\ 1 & 0 & 0 \\ 0 & 0 & -1}).
    \end{equation}
If we define $\rrho_{\pm 1} = \smqty({0 & 1 & 0 \\ 1 & 0 & 0 \\ 0 & 0 & \pm 1})$, then $\rrho_{\pm 1}$ is an invertible Hermitian matrix satisfying $\rrho_{\pm 1}^2 = \Imat$. By repeating this process for the other Gell-Mann matrices we find a total of 15 linearly \emph{dependent} Hermitian matrices, which we arrange with indices ranging from $-7$ to $7$:
    \begin{align}
        \rrho_{\pm a} &= 
        \begin{cases}
            \llambda_a \pm \smqty({0&0&0\\0&0&0\\0&0&1}) & a=1,2,3 \\
            \llambda_a \pm \smqty({0&0&0\\0&1&0\\0&0&0}) & a=4,5 \\
            \llambda_a \pm \smqty({1&0&0\\0&0&0\\0&0&0}) & a=6,7 \\
            \smqty({1&0&0\\0&1&0\\0&0&-1}) & a=0 \\
        \end{cases}.
    \end{align}
The linear combinations of $\rrho_{\pm a}$ build up the Gell-Mann matrices:
    \begin{equation}
        \llambda_a = \begin{cases}
            \tfrac{1}{2} \rrho_{+a} + \tfrac{1}{2} \rrho_{-a} & a=1,2,\ldots, 7 \\
            \tfrac{1}{2 \sqrt{3}} \rrho_{-3} - \tfrac{1}{2 \sqrt{3}} \rrho_{+3} + \frac{1}{\sqrt{3}} \rrho_0 & a =8
        \end{cases}. 
    \end{equation}
Using the invariant decomposition, the exponentials of individual Gell-Mann matrices can easily be calculated. For $\llambda_a$ with $a=1,2,\ldots,7$, we have the exponential factorization
    \begin{align*}
        e^{i \theta \llambda_a} &= e^{i \tfrac{\theta}{2} \rrho_{+a}} e^{i \tfrac{\theta}{2} \rrho_{-a}} \\
        &= \pqty{\Imat \cos(\tfrac{\theta}{2}) + i \rrho_{+a} \sin(\tfrac{\theta}{2})} \pqty{\Imat \cos(\tfrac{\theta}{2}) + i \rrho_{-a} \sin(\tfrac{\theta}{2})} \\
        &= \tfrac{1}{2} (1 + \cos\theta) \Imat - \tfrac{1}{2} (1 - \cos \theta) \rrho_{+a} \rrho_{-a} + i \llambda_a \sin\theta \\
        &= \pqty{\Imat - \llambda_a^2} + \llambda_a^2 \cos \theta + i \llambda_a \sin\theta.
    \end{align*}
This is identical to eq. (7) of \cite{Curtright:2015iba}, as it should.
The $\ccos(i \theta \llambda_a)$ part of these matrices therefore has an equilibrium position 
\[ \pqty{\Imat - \llambda_a^2} = \tfrac{1}{2}\pqty{\Imat - \rrho_{+a} \rrho_{-a}}\] 
when $\theta = \pm \tfrac{\pi}{2}$, about which is being rotated from $\Imat$ at $\theta = 0$, to $- \rrho_{+a} \rrho_{-a}$ when $\theta = \pm \pi$.
The exponential of $\llambda_8$ is most easily calculated from its diagonal form, or using the invariant decomposition and a bit more calculus, as
    \begin{align*}
        e^{i \theta \llambda_a} &= \text{diag}\pqty{e^{i \tfrac{\theta}{\sqrt{3}}}, e^{i \tfrac{\theta}{\sqrt{3}}}, e^{- i \tfrac{2 \theta}{\sqrt{3}}}}.
    \end{align*}
Combined with the identity matrix there are 16 Hermitian matrices $\{ \Imat, \rrho_{\pm a}\}$, which square to $\Imat$, and 16 skew-Hermitian matrices $\{ i \Imat, i \rrho_{\pm a}\}$, which square to $-\Imat$. 
This maps onto the even subalgebra of the geometric algebra $\mathcal{G}(6)$, which C. Doran et al. \cite{LGasSG} proved can be used to describe \SUN{3}. Investigating this link further will be the topic of future research.

\section{Conclusion}
A novel decomposition for $n \times n$ matrices was found. When applying this decomposition to $\Bmat \in \sun{3}$, we found $\Bmat$ could be split into three commuting matrices: $\Bmat = \blmat_1 + \blmat_2 + \blmat_3$,
where each $\blmat_i$ is called \emph{simple}, because its square is $ \lambda_i \Imat$, with $\lambda_i \in \mathbb{R}$.

As the group element $\Umat = \exp\bqty{\Bmat}$ leaves each of the $\blmat_i$ invariant under the transformation $\Umat \blmat_i \Umat^\dagger = \blmat_i$, we named this decomposition the \emph{invariant decomposition} of~$\Bmat$.
We then found that the invariants $\blmat_i$ play an important role in both the exponentials and the logarithms of $\Umat \in \SUN{3}$, as they are both the geometric invariants, and generators, of $\Umat$.

The invariant decomposition offers an easy and intuitive way to perform computations in \SUN{3}, bringing Abelian intuitions into this non-Abelian space. 

\section{Acknowledgements}
The author would like to thank Prof. David Dudal, Prof. Anthony Lasenby, and Steven De Keninck for valuable discussions about this research. Additional gratitude goes to the insights provided by geometric algebra, and \cite{GeometricCalculus,LGasSG} in particular, which were the driving force behind this research.

\printbibliography

\newpage

\end{document}